\documentclass{article}

\usepackage{mathabbrev}
\usepackage{macros}
\usepackage{amssymb}
\usepackage{amsmath}
\usepackage{amsthm}
\usepackage{graphicx}
\usepackage{todonotes}
\usepackage{xcolor}
\usepackage{thm-restate}
\usepackage{hyperref}
\usepackage{lineno}
\usepackage[a4paper, total={5.25in, 8.5in}]{geometry}

\begin{document}
\title{A Set Automaton to Locate All Pattern Matches in a Term}

\author{
Rick Erkens
\\ \texttt{r.j.a.erkens@tue.nl}
\and
Jan Friso Groote
\\ \texttt{J.F.Groote@tue.nl} 
}









\theoremstyle{definition}
\newtheorem{theorem}{Theorem}[section]
\newtheorem{lemma}[theorem]{Lemma}
\newtheorem{corollary}[theorem]{Corollary}
\newtheorem{example}[theorem]{Example}
\newtheorem{definition}[theorem]{Definition}
\newtheorem{proposition}[theorem]{Proposition}

\maketitle

\begin{abstract}
Term pattern matching is the problem of finding all pattern matches in a subject term,
given a set of patterns.
Finding efficient algorithms for this problem is an important direction for research~\cite{har:termindexing}. 
We present a new set automaton solution for the term pattern matching problem
that is based on match set derivatives where each function symbol in the subject pattern is visited exactly once.
The algorithm allows for various traversal patterns over the subject term and 
is particularly suited to search the subject term in parallel. 

\end{abstract}

\section{Introduction}
Given a set of term patterns
and a subject term,
we are interested in the \textit{subterm matching problem}, which is to find all 
locations in the subject term where a pattern matches.
In term rewriting this corresponds to the act of finding all redexes.
Typically, the matching operation must be performed for many subject terms using the same pattern set,
which makes it desirable that matching is efficient. The costs of preprocessing the pattern set is less important as
it is only done once.

The subterm pattern matching problem should not be confused with
the \emph{root (pattern) matching problem}.
In the latter, only the matches at a specific position
in the subject term are needed.
There are many solutions to the root matching problem
that are designed to efficiently deal with sets of patterns \cite{har:termindexing}.
Moreover these solutions have been compared in a the practical setting of
theorem proving \cite{nieuwenhuis:evaluation}.
A solution for the root matching problem can be applied to solve the subterm matching problem
by applying it to every position in a subject term. But this solution can be expensive
as many function symbols in the subject term will be inspected multiple times.

In contrast to the root matching problem,
efficient solutions to the subterm matching problem
are generally restricted to only a single pattern,
and not to a set as is common in term rewriting.
More seriously, they avoid the use of an automaton construction and
process both the pattern and the subject term,
which is expensive if the matching problem needs to be solved for a huge number of subject terms.
Existing solutions for pattern sets are reductions from stringpath matching,
which requires the resulting stringpaths to be merged in order to yield a conclusive answer.
The algorithm that we propose is a mixture of an automaton and the match set approach.
It is explicitly formulated for an arbitrary number of patterns,
operates directly on the subject term in a top-down fashion,
and directly outputs pattern-position pairs instead of stringpath matches.

We present a solution using a so-called \textit{set automaton}.
In a set automaton intermediate
results are stored in a set and these stored results can be processed independently using the same
automaton. This is similar to a pushdown automaton where intermediate results
are stored on a stack to be processed at a later moment. 
A set automaton allows for massive parallel processing. This is interesting
given the prediction that the next boost in computing
comes from developing algorithms that are more parallel in nature \cite{Leisersoneaam9744}. 
 
Given a pattern set $\cL$, we construct a deterministic automaton
that prescribes a traversal of subject terms $t$. The automaton is executed at some
position $p$ in $t$, initially at the root. In each state a next transition is chosen
based on the function symbol $f$ in $t$ at a prescribed position, which is a sub-position of $p$. 
Every function symbol of $t$ is only inspected once.
Each transition is labelled with zero or more outputs of the form $\ell@p'$,
announcing a match of pattern $\ell$ at some position $p'$ in the subject term. 

Each transition ends in a set of next state/position pairs that must be 
processed further. In case the resulting set always consists of one single state/position,
the set automaton behaves as an ordinary automaton. 
The order in which the resulting state/position pairs
need to be processed is undetermined, hence the name \textit{set} automaton. 
In a sequential implementation a stack or queue could be
used to store these pairs giving depth-first or breadth-first strategies. But more interestingly,
the new state/position pairs can be taken up by independent processors, exploring 
the subject term $t$ in parallel. Note that also when running in parallel the algorithm
adheres to its main asset, namely that every function symbol of $t$ will only be inspected
once. 

The set automaton is generated by taking function symbol/position derivatives of match goal sets,
similar to how Brzozowski derivatives work for regular expressions~\cite{brzozowski:derivatives}.
The derivatives are partitioned into independent classes, giving rise to the set of next states. 
By shifting the match goal sets back, the relative displacement through the subject term is derived
allowing to calculate the position where the next state must be evaluated. This keeps the 
automaton finite.

The paper is organized as follows. After some preliminaries we informally discuss an example 
set automaton that matches associativity patterns in Section~\ref{sec:informal}.
Section~\ref{sec:construction} is dedicated to the set automaton construction.
In Sections~\ref{sec:validconstruction} we show that the construction is a well-defined and terminating procedure,
and in Section~\ref{sec:correctevaluation} we prove that the obtained set automaton is indeed a correct and
efficient solution to the subterm matching problem.
In Section~\ref{sec:complexity} we discuss the complexity of applying a set automaton and briefly discuss
some preliminary experiments on the size of set automata.
Lastly in Section~\ref{sec:futurework} we share our thoughts on future work.

\subsection{Related work}
Many solutions for the subterm pattern matching problem
focus on the time complexity or benchmarking of matching \emph{one} pattern against \emph{one} subject term.
See for example \cite{chauve:matching,cole:pattern,travnivcek:modification,dubiner:faster}.
These methods are typically inefficient if there is a large pattern \emph{set},
and the subject terms that need to be matched against the pattern set outnumber 
the subject term size and pattern size by orders of magnitude.
Especially in model checking tools that use term rewriting to manipulate data \cite{toolsetpaper,eker:maude},
the pattern set size is usually a fixed parameter whereas the amount of terms that need to be rewritten blows up
according to state space explosion.
A better solution is to preprocess the pattern set into an automaton-like data structure.
Even though the preprocessing step is usually expensive,
the size of the pattern set is removed as a parameter from the time complexity of the matching time.
This makes the subterm matching problem efficiently solvable against a vast number of subject terms.
To our knowledge, our approach is the first top-down solution of this kind, that achieves this efficiency.

A literature study on related solutions is found in the taxonomy of \cite{cleophas:taxonomies,cleophas:thesis}.
Hoffmann and O'Donnell \cite{hoffmann:matching} convert a pattern into a set of stringpaths,
after which they create an Aho-Corasick automaton \cite{aho:stringmatching} that accepts this set of stringpaths.
Cleophas, Hemerik and Zwaan report that this algorithm is closely related to
their algorithm, which constructs a tree automaton from a single pattern \cite{cleophas:related}.
In \cite{cleophas:thesis}, Algorithm~6.7.9, there is a version of this algorithm that supports multiple patterns.
The disadvantage of both approaches is that a subject term is scanned for matching stringpaths,
rather than term pattern matches.
In order to yield a conclusive answer to the term pattern matching problem,
it is required to keep track which stringpaths match for every pattern, at every position in the subject term.
Our set automata are built directly on the pattern set,
which allows us to output pattern-position pairs directly and avoid the postprocessing step of merging stringpath
matches.

Flouri et al. create a push-down automaton in \cite{flouri:template} from a single pattern.
This approach is very similar to the construction of our set automaton
in the sense that match-sets are used in the automaton construction.
This yields the same complexity as Hoffman and O'Donnell's bottom-up algorithm \cite{hoffmann:matching}.

The notation and the fact that set automaton states are labelled with positions,
have much in common with Adaptive Pattern Matching Automata \cite{sekar:adaptive},
which form a solution to the root pattern matching problem.

\section{Preliminaries}
A signature is a sequence
of disjoint, finite sets of function symbols $\bF_0$,$\bF_1$,\dots,$\bF_n$
where $\bF_i$ consists of function symbols of arity $i$.
We denote the arity of $f$ by $\ar f$.
The set of constants is $\bF_0$,
the entire signature is defined by $\bF=\bigcup_{i=0}^n \bF_i$ and
the set of non-constants is denoted by $\bF_{>0}=\bigcup_{i=1}^n \bF_i$
Let $\bT(\bF)$ be the set of terms over $\bF$, defined as the smallest set that contains the variable $\dc$,
every constant, and for all $f\in\bF_{>0}$,
whenever $t_1,\dots,t_{\#f}\in\bT(\bF)$, then also $f(t_1,\dots,t_{\#f})\in\bT(\bF)$.
The set of closed terms $\bT_C(\bF)$ is defined similarly, but without the clause $\dc\in\bT_C(\bF)$.
Since we only deal with linear patterns,
that is, patterns in which no variable occurs twice,
it is unnecessary to distinguish between the terms $f(x)$ and $f(y)$.
Therefore we only use one variable $\omega$.

A pattern over the signature $\bF$ is a term in $\bT(\bF)\setminus\{\dc\}$.
We use $\ell$ to range over patterns.
A pattern is typically the `left-hand side' of a rewrite rule.
Given a pattern $\ell=f(t_1,\dots,t_n)$, its head symbol is given by $\hd(\ell)=f$.
A pattern set is a finite, non-empty set of patterns.
Throughout this paper we use an arbitrary pattern set denoted by $\cL$.

A position is a list of positive natural numbers.
We use $\bP$ to denote the set of all positions
and we use $\epsilon$ to denote the empty list;
it is referred to as the \emph{root position}.i
Given two positions $p$, $q$ their concatenation is denoted by $p.q$.
The root position acts as a unit with respect to concatenation.

To alleviate the notation, we often denote a pair $(x,p)$ in some set $X\times\bP$ by $x@p$
so that the pair may be read as `$x$ at position $p$'.
The term domain function $\cD:\bT(\bF)\to\cP(\bP)$ maps a term to a set of positions.
That is, $\cD(\dc) = \{\epsilon\}$, for all $a\in\bF_0$ we have $\cD(a)=\{\epsilon\}$, and for all $f\in\bF_n$ with $n>0$
we have $\cD(f(t_1,\dots,t_n))=\{\epsilon\}\cup\bigcup_{i\leq n}\{i.p\mid p\in \cD(t_i)\}$.

Given a term $t$ and a position $p\in\cD(t)$,
the subterm of $t$ at position $p$ is denoted by $t[p]$.
A pattern $\ell$ matches term $t$ on position $p$ iff for all $p'\in\cD(\ell)$ such
that $\ell[p']\neq\dc$ we have that $\hd(t[p.p'])=\hd(\ell[p'])$.

Let $\sub(t)$ be the subpatterns of $t$, given by $\{t[p]\mid p\in\cD(t)\textrm{ and }t[p]\not=\omega\}$.
Since $\dc$ is not a pattern, it is excluded from this set on purpose.
We extend $\cD$ and $\sub$ to sets of terms by pointwise union.
That is, $\cD(\cL)=\bigcup_{\ell\in\cL}\cD(\ell)$, and similarly for $\sub$.

\section{An example set automaton}\label{sec:informal}
In this section we informally discuss the example set automaton in Figure~\ref{fig:associativity}.
It solves the term matching problem for the associativity patterns
$\ell_1=f(f(\dc,\dc),\dc)$ and $\ell_2=f(\dc,f(\dc,\dc))$.
We work in a setting with one binary function symbol $f$ and one constant $a$.

\begin{figure}
\centering
\includegraphics[scale=1.0]{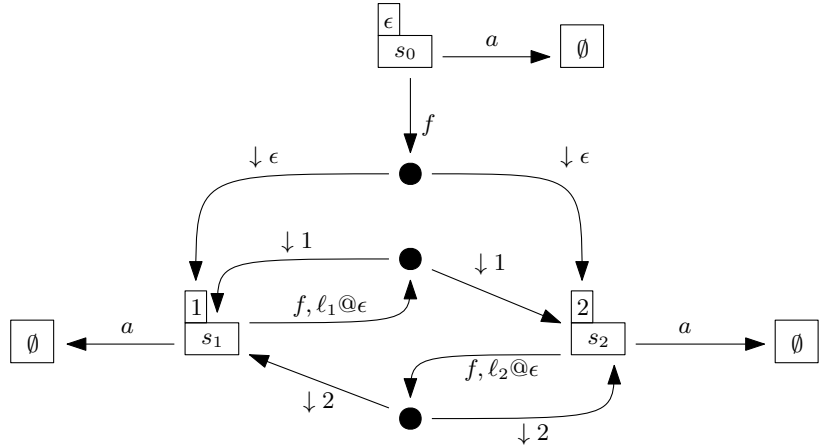}
\caption{A set automaton for the associativity patterns.}
\label{fig:associativity}
\end{figure}

We explain this automaton by applying it to the term $t=f(f(a,f(a,a)),a)$.
The evaluation is done in a semi-top-down fashion.
That is,
in order to inspect position $p.i$ we need to have inspected position $p$ before.
We execute the automaton given a state and a \emph{position pointer} $p$,
which is initially state $s_0$ at the root position.
The automaton tells us which position in $t$ to inspect, which pattern matches are given as an output at which positions,
and it tells at which state/position pairs the evaluation of the automaton must be continued. 

The initial state $s_0$
is labelled with the root position in the box on top of it.
This means that we have to inspect the function symbol in $t$ at position 
$\epsilon$ relative to the position pointer $p$.
Since the position pointer is initially $\epsilon$,
we inspect the head symbol at $t[\epsilon.\epsilon]$ which is $f$.
There are two $f$-transitions from state $s_0$ in the automaton, which have been depicted graphically as an
$f$-labelled arrow, going to a black dot with two outgoing arrows. If a match is found, the transition 
is labelled with $\ell@p'$ to indicate that pattern $\ell$ matches at position $p'$ relative to
the position pointer $p$.
In this case, no such label is present on the $f$-labelled transition.
Therefore no pattern match is reported.
Furthermore, the arrows from the black dots are labelled with a relative displacement $p''$ indicating
that the next state must be evaluated at position pointer $p.p''$.
In this case, the displacement annotation $\downarrow\epsilon$
prescribes that we continue the evaluation at position pointer $\epsilon.\epsilon$.
The two transitions for $f$ go to states $s_1$ and $s_2$ indicating that both states must
be evaluated independently at position $\epsilon$.
This can be done in parallel,
but for simplicity we do a sequential traversal and continue in state $s_1$.

We are in state $s_1$ and the position pointer is still $\epsilon$.
The state label of $s_1$ is $1$, so we look at position $1$ relative to the position pointer.
In term $t=f(f(a,f(a,a)),a)$ we observe $\hd(t[\epsilon.1])=f$,
so we take \emph{both} $f$-transitions from $s_1$.
The arrow labelled by $f$, is accompanied by the label $\ell_1@\epsilon$.
This means that we announce a match for pattern $\ell_1$ at position $\epsilon$
relative to the position pointer.
Since the position pointer is still $\epsilon$, we announce that $\ell_1$ matches $t$ at position $\epsilon$.
From the black dot there are two outgoing arrows with the label $\downarrow 1$.
This means that we continue in states $s_1$ and $s_2$ with the position pointer changed to
$\epsilon.1$. 

Continuing the evaluation in state $s_2$ at position pointer $1$,
we find the state label $2$ on top.
So, we inspect $t$ at position $2$ relative to the position pointer and find that
$\hd(t[1.2])=f$.
We again follow both outgoing $f$-transitions.
First we announce a match for pattern $\ell_2$ at position $\epsilon$ relative
to the position pointer, so we get that $t$ matches $\ell_2$ at position $1$.
Following the arrows from the bottom black dot,
we continue the evaluation in $s_1$ and $s_2$ with position pointer $1.2$.

Now the following state/position pairs still remain to be evaluated: 
$s_2$ at position pointer $\epsilon$,
$s_1$ at $1$, and
$s_1$ and $s_2$ both at position pointer $1.2$.
Inspecting $t$ at each position $p.L(s)$ where $p$ is the position pointer and $L(s)$ is the state label,
we find the constant $a$.
Following any $a$-transition,
the evaluation ends up in the final state, denoted by $\emptyset$, which means that no new
state/positions pairs need to be added for evaluation. 

The algorithm provides the following answer to the question
``at which positions do the patterns $\ell_1=f(f(\dc,\dc),\dc)$ and $\ell_2=f(\dc,f(\dc,\dc))$
match the term $t=f(f(a,f(a,a)),a)$?''. The pattern $\ell_1$ matches $t$ at the root position 
and $\ell_2$ matches $t$ at position $1$.
Observe that the algorithm inspected every position of $t$ exactly once.
The construction of the automaton guarantees this efficiency,
even though at every inspection occurrence of a symbol $f$ two independent 
evaluations of the automaton were started.

\section{Automaton construction}\label{sec:construction}
We describe how to create a set automaton based on \emph{position-/function symbol derivatives}.
To this end we first formally define the automaton,
and in particular, what kind of information should be encoded by states.

The sets of \emph{match obligations} $\MO$ and \emph{match announcements} $\MA$ are respectively defined by
\[\MO=\cP(\sub(\cL)\times\bP)\setminus\{\emptyset\}
\quad
\MA=\cL\times\bP
\,.
\]
A \emph{match goal} is a match obligation paired with a match announcement.
To limit the amount of parentheses, we often denote a match goal,
i.e.\ a pair in $\MO\times \MA$, by $\ell_1@p_1,\dots,\ell_n@p_n\to\ell@p$.
Such a match goal should be read as: ``in order to announce a match for pattern $\ell$ at position $p$,
we are obliged to observe the (sub)pattern $\ell_i$ on position $p_i$, for all $1\leq i\leq n$''.
We denote the positions of a match obligation $mo$ by $\pos(mo)$, defined by $\pos(mo)=\{p\in\bP\mid (t,p)\in mo\}$.

A set automaton for the pattern set $\cL$ is a tuple $(S,s_0,L,\delta,\out)$ where
\begin{itemize}
\item $S\subseteq\cP(\MO\times \MA)\setminus\{\emptyset\}$ is a finite set of states;
\item $s_0\in S$ is the initial state;
\item $L:S\to\bP$ is a state labelling function;
\item $\delta:S\times\bF\to\cP(S\times\bP)$ is a transition function;
\item $\out:S\times\bF\to\cP(\cL\times\bP)$ is an output function.
\end{itemize}
The empty set serves as a final state, but it has no outgoing transitions and no output.
Furthermore, a match goal of the form $\ell@p\to\ell@p$ is called \emph{fresh},
and a match goal of the form $mo\to\ell@\epsilon$ is called a \emph{root goal}.

\begin{example}\label{ex:bigboy}
Consider the pattern $\ell=f(f(\dc,g(\dc)),g(\dc))$.
Figure~\ref{fig:bigboy} is a set automaton for the singleton pattern set $\{\ell\}$.
It serves as a running example throughout this section and the next.
The state labels are given in the small boxes on the top left of every state,
and on the top right of every state there is an identifier.
We have $L(s_0)=\epsilon$ and $L(s_3)=1.2$.
Formally we have $\delta(s_3,f)=\{(s_0,1.1),(s_1,1.2)\}$,
which is depicted graphically as an $f$-labelled arrow going to the black dot,
with two outgoing position-labelled arrows to $s_0$ and $s_2$.
The only non-empty output set is $\out(s_3,g)=\{\ell@\epsilon\}$.
For all other state/symbol pairs $(s,h)$ we have $\out(s,h)=\emptyset$.
The final state $\emptyset$ has two incoming transitions.
For graphical purposes it is displayed twice.
\end{example}

\begin{figure}[h!]
\centering
\includegraphics[scale=0.8]{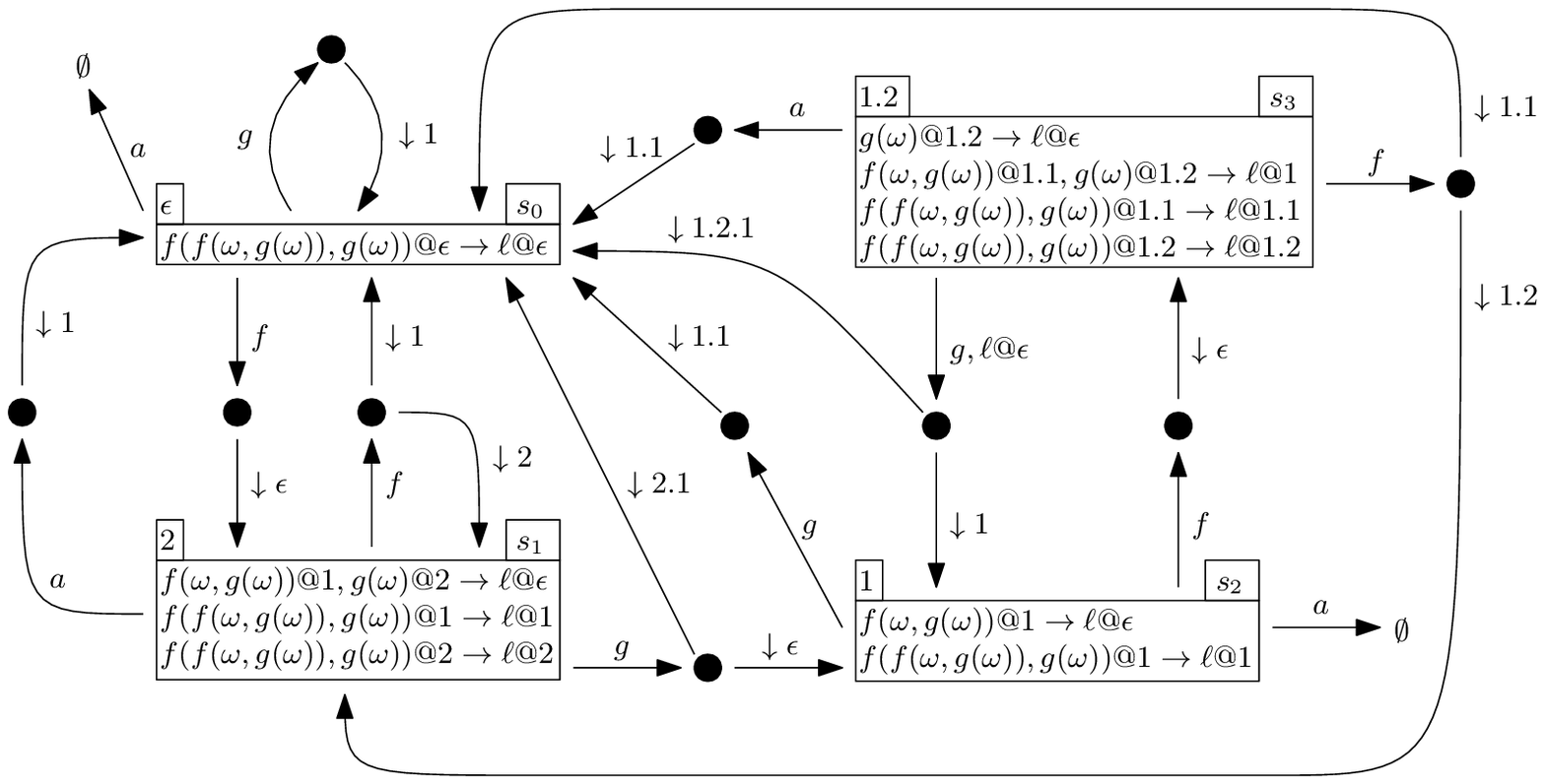}
\caption{A set automaton for $\ell=f(f(\dc,g(\dc)),g(\dc))$.}
\label{fig:bigboy}
\end{figure}

\subsection{Initial state}
Let $\cL$ be a pattern set.
We construct the automaton $M=(S,s_0,L,\delta,\out)$
by starting with the initial state.
It is labelled with the root position and
its match goals are all possible fresh root goals:
\[s_0=\{\ell@\epsilon\to \ell@\epsilon\mid\ell\in\cL\} \quad \text{and} \quad L(s_0)=\epsilon
\,.\]

\subsection{Function symbol-position derivatives}
To determine the transition relation,
we introduce function symbol-position derivatives.
This terminology is borrowed from Brzozowski derivatives of regular expressions~\cite{brzozowski:derivatives}.
From a state $s$ with $L(s)=p$, and a symbol $f$,
we determine the $f$-$p$-derivative of $s$
by computing the reduced match obligations of $s$
and adding the fresh match goal $\ell@p.i\to\ell@p.i$ for every argument $i$ of $f$ and every pattern $\ell\in\cL$.
Based on observing function symbol $f$ at position $p$, the match obligation $\ell_1@p_1,\dots,\ell_n@p_n$
can be altered in one of four ways.
\begin{itemize}
\item $p=p_1$, $n=1$ and $\ell_1=f(\dc,\dots,\dc)$.
Then $f@p$ is the last observation that was needed, so the obligation is fulfilled.
The match announcement paired with this obligation is presented as a pattern match.
\item $p=p_i$ for some $i$ and $\hd(\ell_i)\neq f$. Then $f@p$ contradicts with an expected observation, so the match obligation
is discarded.
\item $p\neq p_i$ for all $i$. Then $f@p$ is unrelated, so the obligation remains unchanged by this observation.
\item otherwise $p=p_i$ for some $i$ and $\hd(\ell_i)=f$,
but $f@p$ is only one of the many expected observations.
Then $\ell_i@p_i$ is removed and the arguments of $\ell_i$ are added as new match obligations.
\end{itemize}
Formally, the mapping $\reduce:\MO\times\bF\times\bP\to\MO\cup\{\emptyset\}$
alters the match obligation $mo$ after the observation $f@p$ by
\begin{align*}
\reduce(mo,f,p)=&
    \{\ell@q\in mo\mid q\neq p\} \ \cup \\
    & \{\ell[i]@p.i\mid\ell@p\in mo\wedge 1\leq i\leq\#f\wedge\ell[i]\neq\dc\}
    \,.
\end{align*}
Using the mapping $\reduce$, we can define the $f$-derivative of state $s$ by
\begin{align*}
\deriv(s,f)&=\mathit{unchanged} \cup \mathit{reduced} \cup \mathit{fresh}\,,\text{where} \\
\mathit{unchanged}&=\{mo\to ma\in s\mid L(s)\notin\pos(mo)\} \\
\mathit{reduced}&=
    \{\reduce(mo,f,L(s))\to ma\mid mo\to ma\in s \ \wedge \\
     & \hspace{2cm} \exists\ell[\ell @L(s)\in mo \wedge \hd(\ell)=f]
     \wedge \reduce(mo,f,L(s))\neq\emptyset\} \\
\mathit{fresh}&=\{\ell@L(s).i\to \ell@L(s).i\mid\ell\in\cL\wedge 1\leq i\leq\#f\}
\end{align*}

\begin{example}\label{ex:bigboyderivative}
Recall the pattern 
$\ell=f(f(\dc,g(\dc)),g(\dc))$
and the set automaton in Figure~\ref{fig:bigboy}.
Consider state $s_1$.
The parts of $\deriv(s_1,g)$ are computed as follows:
\begin{align*}
\mathit{unchanged}&=\{
f(f(\dc,g(\dc)),g(\dc))@1\to \ell@1
\} \\
\mathit{reduced}&=\{
f(\dc,g(\dc))@1\to\ell@\epsilon
\} \\
\mathit{fresh}&=
\{f(f(\dc,g(\dc)),g(\dc))@2.1\to\ell@2.1
\}\,.
\end{align*}
Note that the goal $f(f(\dc,g(\dc)),g(\dc))@2\to\ell@2$ disappears completely
since there is a mismatch with the expected symbol $g$ at position $2$.
\end{example}

\subsection{Derivative partitioning}
One application of $\deriv$ creates new match obligations with strictly lower positions.
Repeated application of $\deriv$ therefore results in an automaton with an infinite amount of states.
To solve this problem we take two more steps after computing the derivative.
First, we partition the derivative into independent equivalence classes.
Then, in every equivalence class, we lower the positions of all match goals as much as possible.
These two measures suffice to create a finite set automaton.

Note from Example~\ref{ex:bigboyderivative} that the derivative has two match obligations at position
$1$, and one match obligation at position $2.1$.
To obtain an efficient matching algorithm,
it is important that goals with overlapping positions stay together to obtain an efficient matching algorithm.
Conversely, sets of goals that are independent from each other can be separated to form a new state with fewer match goals.
When evaluating a set automaton this creates the possibility of exploring parts of the subject term independently.

Given a finite subset of match obligations $X\subseteq\MO$,
define the \emph{direct dependency relation} $R$ on $X$
for all $mo_1,mo_2\in X$ by $mo_1 \mathrel{R} mo_2$, iff $\pos(mo_1)\cap\pos(mo_2)\neq\emptyset$.
Note that $R$ is reflexive (since $\MO$ excludes the empty set) and symmetric.
But $R$ is not transitive, since for the obligations
\[mo_1=\{t_1@1\} \quad mo_2=\{t_1@1, t_2@2\} \quad mo_3=\{t_2@2\}
\]
we have $mo_1\mathrel{R} mo_2 \mathrel{R} mo_3$, but not $mo_1\mathrel{R} mo_3$.
Denote the \emph{dependency relation} on $X$ by $\sim_X$,
defined as the transitive closure of $R$.
Two match obligations are said to be \emph{dependent} iff $mo_1\sim_X mo_2$.
We extend $\dep_X$ to match goals by ${(mo_1\to ma_1)}\dep_X{(mo_2\to ma_2})$ iff $mo_1\dep_X mo_2$.
The subscript $X$ is mostly omitted if the set is clear from the context,
but note that it is necessary to define this relation separately on every state.
Defining it on the set of all match obligations will simply result in the full relation $\MO\times\MO$.

To determine the outgoing transitions
we partition $\deriv(s,f)$ into equivalence classes with respect to dependency $\dep$ on the match obligations.
Each equivalence class then corresponds to a new state.
The set of equivalence classes of the derivative is denoted by $[\deriv(s,f)]_{\dep}$.
We use the letter $K$ to range over equivalence classes.

\begin{example}\label{ex:bigboypartition}
Consider the computed $g$-derivative in Example~\ref{ex:bigboyderivative}.
Partitioning
yields
\begin{align*}
K_1&=
\{
f(f(\dc,g(\dc)),g(\dc))@1\to \ell@1
f(\dc,g(\dc))@1\to\ell@\epsilon
\}\\
K_2&=
\{
f(f(\dc,g(\dc)),g(\dc))@2.1\to\ell@2.1
\}
\end{align*}
\end{example}

\begin{example}\label{ex:bigboypartition2}
Consider the $f$-derivative of $s_2$, which is exactly $s_3$.
Note that the goals $g(\dc)@1.2\to\ell@\epsilon$ and $f(f(\dc,g(\dc)),g(\dc))@1.1\to\ell@1.1$ are not directly dependent,
but the goal $f(\dc,g(\dc))@1.1,g(\dc)@1.2\to\ell@1$ is directly dependent to both goals.
Therefore we obtain a singleton partition.
\end{example}

\subsection{Lifting the positions of classes}
Partitioning into smaller states is not enough to obtain a finite state machine since the positions
of match goals are increasing.
As the last part of the construction, we shorten the positions of every equivalence class.
This can be done due to the following observation.
Suppose that we are looking at term $t$ on position $\epsilon$.
If all match goals say something about position $1$ or lower,
we can remove the prefix $1$ everywhere,
and start to look at term $t$ from position $1$.
Inspecting position $1.p$ from the root is the same as inspecting $p$ from position $1$.

Let $\pos_{\MA}(K)$ denote the positions of the match announcements of $K$.
We want to `lift' every position in every goal of $K$ by the greatest common prefix of $\pos_{\MA}(K)$,
which we denote by $\gcp(\pos_{\MA}(K))$.
To ease the notation we write $\gcp(K)$ instead of $\gcp(\pos_{\MA}(K))$.
Since all positions in a state are of the form $\gcp(K).p'$,
we can replace them by $p'$.
Define $\lift(s)$ by $\lift(s)=\{(\lift(mo),\ell@p')\mid (mo,\ell@\gcp(s).p')\in s\}$
where $\lift(mo)=\{\ell@p'\mid \ell@\gcp(s).p'\in mo\}$.

This concludes the construction of the transition relation.
For a state $s$ and a function symbol $f$,
we fix $\delta(s,f)=\{(\lift(K),\gcp(K))\mid K\in[\deriv(s,f)]_{\dep}\}$.
Note that $\gcp(K)$ is also recorded in each transition since it tells us how to traverse the term.

\begin{example}
Continuing in Example \ref{ex:bigboypartition},
we compute the greatest common prefix and corresponding transition for the two equivalence classes.
For $K_1$ we have $\gcp(K_1)=\gcp(\{1,\epsilon\})=\epsilon$.
Then $\lift(K_1)=K_1=s_2$, and therefore $(s_2,\epsilon)\in\delta(s_1,g)$.
Class $K_2$ has one goal with $\gcp(K_2)=\gcp(\{2.1\})=2.1$.
Then $\lift(K_2)=\{f(f(\dc,g(\dc)),g(\dc))@\epsilon\to\ell@\epsilon\}$,
which yields the transition $(s_0,2.1)\in\delta(s_1,g)$.
\end{example}

\subsection{Output patterns}\label{sec:output}
The output patterns after an $f$-transition
are simply the match announcements that accompany the match obligations
that reduce to $\emptyset$:
\[\out(s,f)=\{ma\in \MA\mid f(\dc,\dots,\dc)@L(s)\to ma\in s\}\,.\]

\begin{example}
Consider state $s_3$ in Figure~\ref{fig:bigboy}.
The goal $g(\dc)@1.2\to\ell@\epsilon$ can be completed upon observing $g$ at position $1.2$,
so we fix $\out(s_3,g)=\{\ell@\epsilon\}$.
\end{example}

\subsection{Position labels}
For every state $s$ there must be a position label $L(s)$ in order to construct the transitions from $s$.
It makes sense to only choose a position from one of the match obligations.
We demand the extra constraint that this position should be part of a root match goal.
The construction guarantees that every state has a root goal,
which we prove in detail in the next section.
Similar to Adaptive Pattern Matching Automata~\cite{sekar:adaptive},
there might be multiple positions available to choose from.
Any of such positions can be chosen in the construction of the automaton,
but this position needs to be fixed when $s$ is created.

\subsection{Summary}
The following is a summary of the construction of the set automaton.
\begin{itemize}
\item $s_0=\{\ell@\epsilon\to \ell@\epsilon\mid\ell\in\cL\}$;

\item $\delta(s,f)=\{(\lift(K),\gcp(K))\mid K\in[\deriv(s.f)]_{\dep}\}$;

\item $\out(s,f)=\{ma\in \MA\mid f(\dc,\dots,\dc)@L(s)\to ma\in s\}$; and

\item $L(s)$ can be any $p\in\pos(mo)$ for some root match goal $mo\to\ell@\epsilon\in s$.
\end{itemize}

\section{Validity of the construction}\label{sec:validconstruction}
In order to see that the construction algorithm of the set automaton works
we need to know whether the following two properties hold.
Firstly, it is necessary that $L(s)$ is a position in the match obligation
of some root goal,
but it is not immediately clear that every state has a root goal.
Secondly, the algorithm needs to terminate.
In this section we show that these properties are valid. 

First we need some extra preliminaries.
In the previous section we used $\gcp(P)$
to denote the greatest common prefix in a set of positions.
This is a lattice construct that requires more elaboration to do proofs.

\begin{definition}[Position join-semilattice]
Position $p$ is said to be below position $q$, denoted by $p\leq q$, iff there is a position $q'$ such that
$p=q.q'$.
Position $p$ is strictly below $q$, denoted by $p<q$, if in addition $q'\neq\epsilon$.
This definition makes the structure $(\bP,\leq)$ a join-semilattice.
That is, $\leq$ is reflexive, transitive and antisymmetric, and
for each finite, non-empty set of positions $P$
there is a unique join $\bigvee P$, which satisfies
$p\leq \bigvee P$ for all $p\in P$ and whenever $p\leq r$ for all $p\in P$ then
also $\bigvee P\leq r$.
We call this join the greatest common prefix $\gcp(P)$.
We denote the join of two positions $p$ and $q$ by $p\vee q$.
Two positions $p,q$ are comparable if $p\leq q$ or $q\leq p$.
\end{definition}

\begin{proposition}\label{prop:positions}
The following properties hold for (sets of) positions.
\begin{itemize}
\item For all $p,q,r\in\bP$ we have $p.q\leq p.r \Leftrightarrow q\leq r$;

\item For all $p\in\bP$, for all $i\in\bN^+$ we have $p\not\leq p.i$;

\item For all $p,q,r\in\bP$, if $p\leq q$ and $p\leq r$ then $q$ and $r$ are comparable;

\item For all $p,q\in\bP$, if $p$ and $q$ are comparable then $p\vee q=p$ or $p\vee q=q$; and

\item For all finite $P,Q\subseteq\bP$ we have $\gcp(P\cup Q)=\gcp(P)\vee\gcp(Q)$.
\end{itemize}
\end{proposition}

Lastly, consider the straightforward notion of reachable state.
A state $s$ is reachable if there is a sequence of transitions to it from $s_0$.
That is, $s_0$ is reachable and
whenever $s$ is reachable and $(s',p)\in\delta(s,f)$,
then $s'$ is also reachable.
The following claims are useful in many places of the correctness proof.

\begin{proposition}\label{prop:matchobligations}
Let $s$ be a reachable state.
\begin{itemize}
\item For all goals $\ell_1@p_1,\dots,\ell_n@p_n\to \ell@p$ in $s$ we have that $p_i\leq p$ for all $i$.
\item For all distinct $p,q\in\pos_{\MO}(s)$ the positions $p$ and $q$ are incomparable.
\item For all distinct $p,q\in\pos_{\MO}(\deriv(s,f))$ the positions $p$ and $q$ are incomparable.
\end{itemize}
\end{proposition}

First, we show that every reachable state always has an available root goal.
By definition of the transition function,
the positions of all match goals in a class $K$ get shortened by $\gcp(K)$
after partitioning.
The partitioning allows us to
show that $\gcp(K)$ is always in $\pos_{\MA}(K)$.

\begin{restatable}{lemma}{rootgoal}\label{lem:rootgoal}
Let $s$ be a reachable state.
Then for all $f\in\bF$,
if $K\in[\deriv(s,f)]_{\dep}$
then there is a goal $mo\to\ell@\gcp(K)$ in $K$.
\end{restatable}

The details of the proof can be found in the appendix;
we give a sketch here.
The proof is by induction on the size of $K$.
The base case is trivial,
and if $|K|\geq 2$ then $K$ can be split into two non-empty classes
with a dependency between them.
By using Propositions~\ref{prop:positions} and~\ref{prop:matchobligations},
and the induction hypothesis
we can show that one of the two smaller classes has
a goal of the right form.

\begin{corollary}
Every reachable state has a root goal.
\end{corollary}

Next, we show that the construction terminates.
There are two key observations to termination.
Firstly, the $\lift$ operation always shortens the positions
of derivative partitions with respect to $\leq$.
Secondly, every state label is a match obligation position of some root goal in that state.
This allows us to prove that reachable states can only have match positions
in some finite set.

\begin{restatable}{lemma}{positionbound}\label{lem:positionbound}
Let $N$ be the largest arity of any function symbol in $\bF$,
and define the set of reachable positions by
$\cR=\{p\in\bP\mid \exists q,r,i: q\in\cD(\cL) \wedge r\in\bP \wedge 1\leq i\leq N \wedge r.p=q.i\}$.
Then for all reachable states $s$ we have that $\pos_{\MO}(s)\subseteq\cR$.
\end{restatable}
The proof can be found in the appendix.
Intuitively, since there are only finitely many state labels,
the longest position in any match obligation is of the form $L(s).i$
where $i$ is bounded by $N$.

\begin{corollary}
There are finitely many reachable states.
\end{corollary}

\section{Correctness of the evaluation}\label{sec:correctevaluation}
The informal evaluation that was discussed in Section~\ref{sec:informal} describes how to apply an automaton $M$ to a
subject term.
Formally this procedure can be defined by the mapping
$\eval_M:S\times\bP\times\bT(\bF)\to\cP(\cL\times\bP)$ given by
\[\eval_M(s,p,t)=\{\ell@p.q\mid \ell@q\in\out(s,f)\}\cup\bigcup_{(s',p')\in\delta(s,f)}\eval(s',p.p',t)\]
where $f=\hd(t[p.L(s)])$.
Finding all pattern matches in a term $t$ is the invocation of $\eval_M(s_0,\epsilon,t)$.
The desired correctness property can then be stated as follows:
\[
\eval_M(s_0,\epsilon,t)=\{\ell@p\in\cL\times\bP\mid\text{$\ell$ matches $t$ at $p$}\}
\,.
\]
This property cannot be shown by a straightforward structural induction on $t$.
In this section we take a detour and prove an equivalent correctness claim.
The proof is sketched as follows.
First, we add explicit structure to the evaluation by computing an evaluation tree $ET_M(t)$ of a term $t$.
We prove a one-to-one correspondence between the nodes of $ET_M(t)$ and $t$.
It follows that this method of pattern matching is efficient
in the sense that every position of $t$ is inspected exactly once.
Soundness and completeness is shown at the end of the section.

\subsection{Evaluation trees}
\begin{definition}
An \emph{evaluation tree} for an automaton $M=(S,s_0,L,\delta,\out)$ is a tuple $(N,\rightarrow)$
where $N\subseteq S\times\bP$ is a set of nodes,
and ${\rightarrow}\subseteq N\times N$ is a set of directed edges.
With a closed term $t$ we associate an evaluation tree $ET_M(t)=(N,\rightarrow)$ defined as the smallest
evaluation tree such that
\begin{itemize}
\item there is a root $(s_0,\epsilon)\in N$; and
\item whenever $(s,p)\in N$ and $\hd(t[p.L(s)])=f$ then
for every $(s',p')\in\delta(s,f)$ there is an edge $(s,p)\to(s',p.p')$ with
$(s',p')\in N$.
\end{itemize}
The successors of a node $n$ are given by $\suc(n)=\{n'\in N\mid n\rightarrow n'\}$.
\end{definition}

\begin{example}\label{ex:evaltree}
Figure~\ref{fig:evaltree} shows the term $t=f(g(a),f(f(a,g(a)),g(a)))$ and
its evaluation tree $ET_M(t)$, given the set automaton $M$ of Figure~\ref{fig:bigboy}.
There is a one-to-one correspondence between the positions of $t$ and the nodes of the evaluation tree.
\end{example}
\begin{figure}
\centering
\includegraphics[scale=1.0]{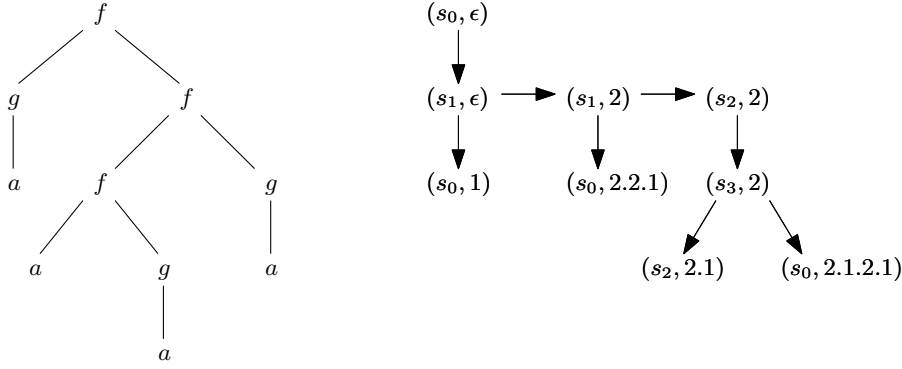}
\caption{The term $t=f(g(a),f(f(a,g(a)),g(a)))$ on the left and its evaluation tree $ET_M(t)$ on the right.}
\label{fig:evaltree}
\end{figure}
We prove that $ET_M(t)$ indeed corresponds to $t$ in general.
To this end, we define for every node the set of positions that still has to be inspected.
That is, the set of work that still has to be done.

\begin{definition}
Define the mapping $\cW:N\to\cP(\cD(t))$ by
\[
\cW(s,p)=\{p.q\in\cD(t)\mid \exists r:r\in\pos_{\MO}(s) \wedge q\leq r\}\,.
\]
\end{definition}
By definition of $s_0$ we have $\cW(s_0,\epsilon)=\cD(t)$.
Intuitively this makes sense, since at the beginning of the evaluation, no work is done and
all the positions still have to be inspected.
The mapping $\cW$ fixes a correspondence between an evaluation tree and the strict subset ordering $(\cD(t),\subset)$.
This follows from the following lemma.
A detailed proof can be found in the appendix.

\begin{restatable}{lemma}{orderingisomorphism}\label{lem:orderingisomorphism}
Let $ET_M(t)=(N,\rightarrow)$
and consider an arbitrary node $(s,p)\in N$.
Then
\begin{enumerate}
\item \label{lem:evaltreedecreasing}
For all successors $(s',p.p')\in\suc(s,p)$ we have that $p.L(s)\notin\cW(s',p.p')$.

\item \label{lem:evaltreesplitting}
For all distinct successors $(s_1,p.p_1),(s_2,p.p_2)\in\suc(s,p)$
the sets $\cW(s_1,p.p_1)$ and $\cW(s_2,p.p_2)$ are disjoint.

\item \label{lem:evaltreehierarchy}
We have that $\cW(s,p)=\{p.L(s)\}\cup\bigcup_{n\in \suc(s,p)} \cW(n)$.
\end{enumerate}
\end{restatable}

By combining these properties,
we get the following two corollaries.

\begin{corollary}
For all terms $t$, we have that $ET_M(t)=(N,\rightarrow)$ is a finite tree.
\end{corollary}
\begin{corollary}\label{cor:bijection}
Define $\varphi:N\to\cD(t)$ by $\varphi(s,p)=p.L(s)$.
Then $\varphi$ is a bijection.
\end{corollary}

It follows that the evaluation of a term terminates,
and every position is inspected exactly once.
Whenever an evaluation tree node has multiple outgoing edges, it means that parallellism is possible.
This parallellism preserves the efficiency of no observation being made twice.

\subsection{Soundness and completeness}
First, consider the following evaluation function that takes an evaluation tree node
and traverses it until a leaf node is reached.
\begin{definition}
Given $ET_M(t)=(N,\rightarrow)$,
define $\eval_M:N\to\cP(\cL\times\bP)$ by
\[
\eval_M(s,p)=\{\ell@p.q\mid \ell@q\in\out(s,\hd(t[p.L(s)]))\}\cup\bigcup_{(s',p.p')\in\suc(s,p)}\eval_M(s',p.p')
\,.
\]
\end{definition}
By Corollary~\ref{cor:bijection}, applying $\eval$ on the initial state from the root position is the same as retrieving the
output at every level of the evaluation tree.
\begin{equation}\label{eq:treetraversal}
\eval_M(s_0,\epsilon)=\bigcup_{(s,p)\in N}\{\ell@p.q\mid \ell@q\in\out(s,\hd(t[p.L(s)]))\}
\,.
\end{equation}

\begin{restatable*}[Correctness]{theorem}{correctness}
For all closed terms $t$,
\[\eval_M(s_0,\epsilon)=\{\ell@p\in\cL\times\cD(t)\mid\text{$\ell$ matches $t$ at $p$}\}
\,.
\]
\end{restatable*}
We show both inclusions at the end of this section.
The inclusion from left to right is the soundness claim.
When the evaluation yields an output,
then it is indeed a correct match.
The inclusion from right to left is the completeness claim.
When some pattern matches at some position,
then the evaluation will output it at some point.

To understand soundness, consider that match goals carry history.
Intuitively, a match goal $a@1, b@2\to f(a,b)@\epsilon$
has a history of having seen $f$ already.
A state with this goal can only be reached by evaluating a term with symbol $f$.
This notion can be formalised as follows.

\begin{definition}
The \emph{history of an evaluation tree node $(s,p)$ respects $t$} iff
for all goals $mo\to\ell@q\in s$,
for all $r\in\cD(\ell)$ such that $\ell[r]\neq\dc$,
if there is some $r'\in\pos(mo)$ with $r\not\leq r'$ then
$\hd(t[p.r])=\hd(\ell[r])$.
\end{definition}
With this definition, the following invariant is the key to soundness.
A proof can be found in the appendix.

\begin{restatable}{lemma}{history}\label{lem:history}
Let $ET_M(t)=(N,\rightarrow)$.
The history of every node $(s,p)$ respects $t$.
\end{restatable}

To understand completeness, observe that upon taking derivatives
a fresh match obligation is added for every new position.
The partitioning then takes care of grouping the fresh goals with
other goals that have the same positions.

\begin{proposition}\label{prop:allfresh}
Whenever a state has a match obligation on position $p$,
then it has the fresh match goal $\ell@p\to\ell@p$ for all $\ell\in\cL$ as well.
\end{proposition}

The following invariant connects to Proposition~\ref{prop:allfresh}.
Intuitively, if a term matches pattern $\ell$ at position $p.q$,
and the evaluation tree reaches a state with some goal $mo\to\ell@q$ is a match announcement,
then this announcement belongs to some goal in some state visited by $\eval$,
until it is given as an output. A detailed proof can be found in the appendix.

\begin{restatable}{lemma}{followthemo}\label{lem:followthemo}
If $\ell$ matches $t$ at $p.q$ and there is a node $(s,p)$ and a match goal
$mo\to\ell@q\in s$ then either $\ell@q\in\out(s,\hd(t[p.L(s)])$ or
there is a node $(s',p.p')\in\suc(s,p)$ such that $s'$ has some goal
$mo'\to\ell@q'$ with $q=p'.q'$.
\end{restatable}

\correctness
\begin{proof}
As mentioned before, we show soundness and completeness.
\begin{itemize}
\item[$\subseteq$]
By Equation~\ref{eq:treetraversal} it suffices to show that for all nodes $(s,p)$,
whenever $\ell@q\in\out(s,\hd(t[p.L(s)]))$ then $\ell$ matches $t$ at $p.q$.
Consider that $\hd(t[p.L(s)])=f$.
By definition of $\eta$, see Section~\ref{sec:output}, we have $f(\dc,\dots,\dc)@L(s)\to\ell@q\in s$.
By Lemma~\ref{lem:history}, the history of node $(s,p)$ respects $t$.
Then for all positions $r\in\cD(\ell)$ with $\ell[r]\neq\dc$ and $r\neq L(s)$ we have that
$\hd(t[p.r])=\hd(\ell[r])$.
From the additional observation $\hd(\ell[L(s)])=f=\hd(t[p.L(s)])$
and Proposition~\ref{prop:matchobligations}
it follows that $\ell$ matches $t$ at $p.q$.

\item[$\supseteq$]
Consider that $\ell$ matches $t$ at $p$.
By Corollary~\ref{cor:bijection},
consider the node $\varphi^{-1}(p)=(s,q)$.
By definition of $\varphi$ we have $q.L(s)=p$.
Since $L(s)\in\pos_{\MO}(s)$,
the fresh goal $\ell@L(s)\to\ell@L(s)$ is $s$ by Proposition~\ref{prop:allfresh}.
Then the repeated application of Lemma~\ref{lem:followthemo}
yields a node $(s',q.q')$ such that $s'$ has some goal $mo'\to\ell@r$
with $L(s)=q'.r$ and $\ell@r\in\out(s',\hd(t[q.q'.L(s')]))$.
Then $\ell@q.q'.r\in\eval(s',q.q')$ by definition of $\eval$.
Since $q.q'.r=q.L(s)=p$ it follows that $\ell@p\in\eval(s',q.q')$.
By Equation~\ref{eq:treetraversal} we conclude $\ell@p\in\eval(s_0,\epsilon)$.
\end{itemize}
\end{proof}

\section{Complexity and automaton size}\label{sec:complexity}
Given an automaton $M$ of pattern set $\cL$,
the matching algorithm $\eval_M(s_0,t)$ runs in $O(d(n+m))$ time where
$n$ is the number of function symbols in $t$, and $m$ is the amount of pattern matches in $t$,
and $d$ is the maximal depth of any pattern in $\cL$.
The factor $d$ is due to the fact that observing a function symbol on position $L(s)$
takes $|L(s)|$ time in general.

The size of a set automaton is exponential in the worst case,
which is not surprising due to similar observations concerning the root pattern matching problem.
Gr\"af observed that a left-to-right pattern matching automaton is exponentially large
in the worst case \cite{graf:lefttoright}.
Sekar et al. observed that adaptive pattern matching automata are exponentially big in the worst case as well,
although a good traversal can reduce the automaton size exponentially in some cases \cite{sekar:adaptive}.

However, practical experiments with pattern sets
show that the automaton size is small, which is in line with other forms of automaton based matching.
We generated set automata to match the left hand sides of rewrite systems used in mCRL2 \cite[Appendix B]{DBLP:books/mit/GrooteM2014},
see Table~\ref{tab:mcrl2experiments}.
In almost all cases the amount of states in the set automaton does not exceed the number of patterns.

\begin{table}
\centering
\begin{tabular}{l|lll}
Specification & Signature size & Amount of patterns & Amount of states \\ \hline
int&    22& 50&27\\ \hline 
pos&    15 & 46&45\\ \hline 
nat&    37&91&117\\ \hline 
fset&   15&28&23\\ \hline 
set&    20&40&24\\ \hline 
list&   16&26&24\\ \hline 
bool&   9&27&14\\ \hline 
bag&    29&44&32\\ \hline 
fbag&   18&30&25\\ \hline 
real&   30&31&31\\ \hline 
\end{tabular}
\caption{The set automaton sizes for parts of the default mCRL2 specification}
\label{tab:mcrl2experiments}
\end{table}

The degree of freedom in the choice of state labels strongly influences the set automaton size.
Consider for example the set of terms $\{t_n\}_{n\in\bN}$ given by $t_0=\dc$ and $t_{n+1}=f(t_n,g(\dc))$.
The set automaton in Example~\ref{ex:bigboy} is generated for pattern set $\{t_2\}$.
We found that the choice of state labels influences the automaton size by a quadratic factor.
By choosing the right-most available position one obtains an automaton of size $2n$
for the pattern set $\{t_n\}$.
A left-most strategy yields an automaton of size $n^2+n$ for $\{t_n\}$.

\section{Future work}
\label{sec:futurework}
The original motivation for this work is to construct a high performance term
rewriter suited for parallel processing, which can both work on a single large term 
as well as on many small terms, repeatedly. This means that the matching
effort must be minimal, which is provided by the automaton, and it also requires 
that the subject term is not transformed before matching commences.
To enable term rewriting, our matching algorithm must still be extended with term rewriting along
lines set out in \cite{hoffmann:interpreter}.
We want to employ that we know the structure of the right-hand side of a rewrite step,
minimizing inspecting known parts of a newly constructed term.
Fokkink et al. have a similar approach in \cite{fokkink:arm}, based on Hoffmann and O'Donnell's
algorithm from \cite{hoffmann:matching}.

Our algorithm has freedom in the position of the function symbol to
be selected, as well as in the next state/position pair that the evaluator chooses. It is interesting
to see whether with knowledge about the distribution of function symbols in subject terms, this freedom 
can be exploited to construct a most efficient set automaton. For instance, we may want to generate the
first match as quickly as possible. This is particularly interesting in 
combination with rewriting where some sub-terms do not have to be inspected as they will be removed
by the rewriting rules.

Observe that the algorithm as it stands does not employ non-linear patterns in line
with matching algorithms such as \cite{sekar:nonlinear}.
But in term rewriting non-linear patterns do occur and therefore an extension to support them is desired.
An extension that provides all matches in a setting where some symbols are known to be associative and/or commutative
would also be interesting.

\bibliographystyle{plain}
\bibliography{bibliography}

\appendix

\section{Appendix}
\subsection{Proof for Lemma~\ref{lem:rootgoal}}
\rootgoal*
\begin{proof}
By induction.
If $|K|=1$ then the claim follows trivially.
If $|K|\geq 2$,
then by virtue of $\dep$ being a transitive closure,
we can partition $K$ into subsets $K_1, K_2\subseteq K$ such that
\begin{itemize}
\item there are match goals $mo_1\to \ell_1@p_1\in K_1$ and 
$mo_2\to\ell_2@p_2\in K_2$
such that $\pos(mo_1)\cap\pos(mo_2)\neq\emptyset$; and
\item partitioning $K_1$ and $K_2$ with respect to $\dep_{K_1}$ and $\dep_{K_2}$ respectively yields
$K_1$ and $K_2$.
\end{itemize}
Pick a position $p$ with $p\in\pos(mo_1)$ and
$p\in\pos(mo_2)$.
By Proposition~\ref{prop:matchobligations} we know that $p\leq p_1$ and $p\leq p_2$.
By the induction hypothesis there are two match goals
$mo_1'\to\ell_1'@\gcp(K_1)\in K_1$ and
$mo_2'\to\ell_2'@\gcp(K_2)\in K_2$.
By properties of $\gcp$,
it follows that $p_1\leq\gcp(K_1)$ and $p_2\leq\gcp(K_2)$.
By transitivity we get $p\leq\gcp(K_1)$ and $p\leq\gcp(K_2)$.
Then by Proposition~\ref{prop:positions},
$\gcp(K_1)$ and $\gcp(K_2)$ are comparable.
Since 
\[
\gcp(K)
= \gcp(\pos_{\MA}(K_1)\cup\pos_{\MA}(K_2))
= \gcp(\pos_{\MA}(K_1))\vee\gcp(\pos_{\MA}(K_2))
\,,
\]
by adding the syntactic sugar for $\gcp(K_1)$ and $\gcp(K_2)$,
and by applying position properties,
it follows that $\gcp(K)=\gcp(K_1)$ or $\gcp(K)=\gcp(K_2)$.
Since both match goals
$mo_1'\to\ell_1'@\gcp(K_1)\in K_1$ and
$mo_2'\to\ell_2'@\gcp(K_2)\in K_2$ are in $K$, we conclude the proof.
\end{proof}

\subsection{Proof for Lemma~\ref{lem:positionbound}}
\positionbound*
\begin{proof}
The initial state easily satisfies the claim.
We show that the claim is an invariant over the production of a transition.

Let $s$ be a reachable state and suppose that $\pos_{\MO}(s)\subseteq\cR$.
Let $f\in\bF$ and consider that $(s',p')\in\delta(s,f)$.
By definition $s'=\lift(K)$ and $p'=\gcp(K)$ for some $K\in[\deriv(s,f)]_{\dep}$.
We have to show that $\pos_{\MO}(\lift(K))\subseteq\cR$.

Consider some position $p\in\pos_{\MO}(\lift(K))$.
By definition of $\deriv$ and $\lift$,
we have that $\gcp(K).p\in\pos_{\MO}(K)$.
Observe that $\cR$ is upward closed under the position prefix ordering $\leq$.
That is, whenever $x\in\cR$ and $x \leq y$ then $y\in\cR$.
Therefore we can ignore $\gcp(K)$; it suffices to show that $p\in\cR$.

If $p$ is the position of an unchanged pair in some match obligation of $K$,
then $p\in\cR$ by assumption.
If $p$ is a position in a changed pair of some fresh or reduced match obligation,
then it suffices to show that $L(s).i\in\cR$ for all $i\leq\# f$.
By construction, $L(s)$ is the position of a root goal in $s$.
Therefore, $L(s)\in\cD(\cL)$.
Since $\#f\leq N$, we have that $i\leq N$ as well.
Hence, $L(s).i\in\cR$.
\end{proof}

\subsection{Proof for Lemma~\ref{lem:orderingisomorphism}}
\orderingisomorphism*
\begin{proof}
Consider that $\hd(t[p.L(s)])=f$.
By construction of $M$ and $ET_M(t)$,
we can characterise the successors of node $(s,p)$ by
\begin{equation}\label{eq:successors}
\suc(s,p)=\{(\lift(K),p.\gcp(K))\in S\times\bP\mid K\in[\deriv(s,f)]_{\dep}\}\,.
\end{equation}

\begin{enumerate}
\item
Towards a contradiction,
using Equation~\ref{eq:successors},
pick an equivalence class $K\in[\deriv(s,f)]_{\dep}$ and
assume that $p.L(s)\in\cW(\lift(K),p.\gcp(K))$.
By definition of $\cW$, there is a pair
$\ell@p$ in some match obligation in $\lift(K)$ and some $q\leq p$ such that
$p.L(s)=p.\gcp(K).q$.
From the position properties it follows that $L(s)=\gcp(K).q$.
From $q\leq p$ it follows that $\gcp(K).q\leq\gcp(K).p$.

Since $\ell@p$ is part of a match obligation in $\lift(K)$,
by definition $\ell@\gcp(K).p$ is part of a match obligation in $K$.
Since $K\subseteq\deriv(s,f)$,
there are two possibilities.
\begin{itemize}
\item If $\ell@\gcp(K).p=\ell@L(s).i$ then it is part of a reduced or fresh match goal.
Then $\gcp(K).p=L(s).i$ for some index $1\leq i\leq \# f$.
But then by
\[L(s)=\gcp(K).q\leq\gcp(K).p= L(s).i\,,\]
we have $L(s)\leq L(s).i$, which contradicts Proposition~\ref{prop:positions}.
\item Otherwise $\ell@\gcp(K).p$ is also part of a match obligation in $s$.
But since $L(s)\leq\gcp(K).p$ and $L(s)\in\pos_{\MO}(s)$,
it must be that $L(s)=\gcp(K).p$ by Proposition~\ref{prop:matchobligations}.
Then, by definition of $\reduce$ it cannot be that $\ell@\gcp(K).p$ is a match obligation of $K$,
a contradiction.
\end{itemize}

\item
By Equation~\ref{eq:successors}, let $K_1,K_2\in[\deriv(s,f)]_{\dep}$ such that $s_1=\lift(K_1)$ and $s_2=\lift(K_2)$,
and $p_1=\gcp(K_1)$ and $p_2=\gcp(K_2)$.

Towards a contradiction,
pick a position $q$ such that that $q\in\cW(\lift(K_1),p.\gcp(K_1))$ and $q\in\cW(\lift(K_2),p.\gcp(K_2))$.
By definition of $\cW$ there are pairs $\ell_1@q_1$ and $\ell_2@q_2$ that are part of some match obligation in
$\lift(K_1)$ and $\lift(K_2)$ respectively, and there are two positions $q_1'\leq q_1$ and $q_2'\leq q_2$
such that $q=p.\gcp(K_1).q_1'$ and $q=p.\gcp(K_2).q_2'$.
Then it follows that $\gcp(K_1).q_1'=\gcp(K_2).q_2'$,

By definition of $\lift$, the pairs $\ell_1@\gcp(K_1).q_1$ and $\ell_2@\gcp(K_2).q_2$
are part of some match obligation in $K_1$ and $K_2$ respectively.
But then from $\gcp(K_1).q_1'\leq \gcp(K_1).q_1$ and $\gcp(K_1).q_1'\leq\gcp(K_2).q_2$
it must be that $\gcp(K_1).q_1$ and $\gcp(K_2).q_2$ are comparable.
Since $\ell_1@\gcp(K_1).q_1$ and $\ell_2@\gcp(K_2).q_2$ are both
elements of $\deriv(s,f)$,
by Proposition~\ref{prop:matchobligations} it follows that
$\gcp(K_1).q_1=\gcp(K_2).q_2$,
which violates the assumption that $K_1$ and $K_2$ are distinct equivalence classes.

\item
Let $\hd(t[p.L(s)])=f$.
By Equation~\ref{eq:successors} we should show that
\[
\cW(s,p)=\{p.L(s)\}\cup\bigcup_{K\in[\deriv(s,f)]_{\dep}}\cW(\lift(K),p.\gcp(K))
\,.
\]
We prove both inclusions.
\begin{itemize}
\item[$\supseteq$]
For the singleton set,
it follows from $L(s)\in\pos_{\MO}(s)$ and the definition of $\cW$ that $p.L(s)\in\cW(s,p)$.
For the big union,
consider some $K\in[\deriv(s,f)]_{\dep}$
and a position $p.\gcp(K).q\in\cW(\lift(K),p.\gcp(K))$.
By definition of $\cW$ there is a pair $\ell@r$ which is part of some match obligation in $\lift(K)$
such that $q\leq r$.
Then $\ell@\gcp(K).r\in mo'$ with $mo'$ a match obligation in $K$.

From $K\in[\deriv(s,f)]_{\dep}$ there are two cases.
If $\ell@\gcp(K).r$ is in some match obligation in $s$,
then $p.\gcp(K).q\in\cW(s,p)$ by virtue of $\gcp(K).q\leq r$ and $r\in\pos_{\MO}(s)$.
Otherwise,
$\gcp(K).r=L(s).i$ for some $i\leq\#f$ and $\ell@\gcp(K).r$ is part of a fresh or reduced match obligation.
Since $L(s)\in\pos_{\MO}(s)$ there is a pair $\ell'@L(s)$ in $s$.
Then $p.\gcp(K).q\in\cW(s,p)$ because $\gcp(K).q\leq\gcp(K).r=L(s).i\leq L(s)$.

\item[$\subseteq$]
Let $mo$ be a match obligation in $s$, let $\ell@r\in mo$ and
consider a position $q$ with $q\leq r$.
We have to show that $p.q\in\{p.L(s)\}$ or there is a $K\in[\deriv(s,f)]_{\dep}$
with $p.q\in\cW(\lift(K),p.\gcp(K))$.
It suffices to distinguish two cases.
\begin{itemize}
\item
In the case $L(s)\neq r$,
then $\ell@r$ is a pair in some match obligation $mo$ in $\deriv(s,f)$.
Then there is an equivalence class $K$ such that $\ell@r$
is in some match obligation of $K$.
Then $r=\gcp(K).r'$ for some $r'$
and $\ell@r'$ is in the match obligation $\lift(mo')$ of the state $\lift(K)$.

We have to show that $p.q\in\cW(\lift(K),p.\gcp(K))$.
From $q\leq r$ and $r=\gcp(K).r'$ we get that $q=\gcp(K).r'.r''$
for some $r''$.
Then $p.q=p.\gcp(K).r'.r''$.
By definition of $\cW$ and from $\ell@r'$ being a match obligation in $\lift(K)$,
it follows that $p.q\in\cW(\lift(K),p.\gcp(K))$.

\item
In the case $r=L(s)$
then $p.q\leq p.L(s)$.
If $q=r=L(s)$ then $p.q=p.L(s)$,
which is in the singleton set $\{p.L(s)\}$.
Otherwise, $q<L(s)$.
Then there is an index $i$ such that $q\leq L(s).i$.
Since $p.q\in\cD(t)$ and $\hd(t[p.L(s)])=f$ it must be that
$i\leq\#f$.
Then by definition of $\deriv(s,f)$
there is a fresh match obligation $\ell@L(s).i\to\ell@L(s).i$ in $K$.
By definition of $\lift$ we have that $\gcp(K).r=L(s).i$ for some $r$
and $\ell@r$ is a match obligation in $\lift(K)$.
Then the proof obligation follows by
$p.q\leq p.L(s).i = p.\gcp(K).r\in\cW(\lift(K),\gcp(K))$.
\end{itemize}
\end{itemize}
\end{enumerate}
\end{proof}

\subsection{Proof for Lemma~\ref{lem:history}}
\history
\begin{proof}
The history of $(s_0,\epsilon)$ trivially respects $t$.
Consider a node $(s,p)$ whose history respects $t$,
and let $f=\hd(t[p.L(s)])$.
Consider a successor $(\lift(K),p.\gcp(K))\in\suc(s,p)$ for some $K\in[\deriv(s,f)]_{\dep}$.
We show that the history of $(\lift(K),p.\gcp(K))$ respects $t$ as well.

Following the definition of $\deriv$, we only look at the reduced match goals in $\lift(K)$.
By definition those are match goals $mo\to ma$ with some pair $f(t_1,\dots,t_n)@L(s)$.
The history of the unchanged goals respects $t$ by assumption and fresh match goals have no history.
Suppose that $mo'\to \ell@p$ is a reduced match goal.
Then $mo'=\reduce(mo,f,L(s))$ with $mo\to ma\in s$.
By definition of $\reduce$ we have that
\[mo'=\{\ell@q\in mo\mid q\neq L(s)\} \cup \{\ell[i]@L(s).i\mid\ell@L(s)\in mo\wedge 1\leq i\leq\#f\wedge\ell[i]\neq\dc\}.\]

For all unchanged pairs $\ell@q$ with $q\neq L(s)$
we do not have to prove anything.
If $\ell'@L(s)$ is a pair in $mo$ then
$\ell'[i]@L(s).i$ is a pair in $mo'$ for all $i$ with $\ell'[i]\neq\dc$.
Hence, $L(s)\in\cD(\ell)$, $\ell[L(s)]\not\leq r$
and $L(s)\not\leq r$ for all goals $\ell'@p\in mo'$.
So, $\hd(t[p.L(s)])=\hd(\ell[L(s)])$.
\end{proof}

\subsection{Proof of Lemma~\ref{lem:followthemo}}
\followthemo*
\begin{proof}
Suppose that $\hd(t[p.L(s)])=f$.
We distinguish two cases.
\begin{itemize}
\item
If $mo=\{f(\dc,\dots,\dc)@L(s)\}$ then
$\reduce(mo,f,L(s))=\emptyset$.
By construction $\ell@q\in\out(s,\hd(t[p.L(s)])$,
as needed to conclude.

\item
Otherwise, let $mo'=mo$ if $L(s)\notin\pos(mo)$ and
$mo'=\reduce(mo,f,L(s))$ if $L(s)\in\pos(mo)$.
Note that if $L(s)\in\pos(mo)$ then $\reduce(mo,f,L(s))$ is not empty.
Then $mo'\to\ell@q\in K$ for some $K\in[\deriv(s,f)]_{\dep}$.
By construction $\lift(mo')\to\ell@q'\in\lift(K)$ for some $q'$ such that
$q=\gcp(K).q'$.
Then $(\lift(K),p.\gcp(K))$ is the node that we are looking for. 
\end{itemize}
\end{proof}

\end{document}